\documentclass[10pt]{extarticle}
\usepackage{graphicx} 
\usepackage{mathtools, amsmath, amssymb, amsthm}
\usepackage{enumitem}
\usepackage{makecell} 

\usepackage{caption} 
\usepackage{subcaption}
\usepackage{xcolor}         
\usepackage[utf8]{inputenc}
\usepackage[authoryear]{natbib}
\usepackage[colorlinks=true,
            linkcolor=blue,  
            citecolor=black,  
            urlcolor=blue]   
           {hyperref}
\usepackage[nameinlink,capitalize,noabbrev]{cleveref} 

\newcommand{\eps}{\varepsilon}

\newcommand{\E}{\mathbb{E}}

\newcommand{\defeq}{\vcentcolon=}   

\newcount\Comments  
\definecolor{darkgreen}{rgb}{0,0.6,0}
\newcommand{\kibitz}[2]{\ifnum\Comments=1{\color{#1}{#2}}\fi}

\newtheorem{theorem}{Theorem}
\newtheorem*{theorem_restatement}{Theorem}

\newtheorem{lemma}[theorem]{Lemma}
\newtheorem*{lemma_restatement}{Lemma}

\newtheorem{corollary}{Corollary}
\newenvironment{proofsketch}{\begin{proof}[Proof sketch]}{\end{proof}}

\title{Strategyproof Facility Location for Five Agents on a Circle using PCD}
\author{Ido Farjoun, Technion\\ Advisor: Reshef Meir}

\date{October 2025}
\usepackage[margin=1.3in]{geometry}

\begin{document}

\maketitle

\begin{abstract}
\noindent We consider the strategyproof facility location problem on a circle. We focus on the case of 5 agents, and find a tight bound for the PCD strategyproof mechanism, which selects the reported location of an agent in proportion to the length of the arc in front of it. We  methodically "reduce" the size of the instance space and then use standard optimization techniques to find and prove the bound is tight. Moreover we hypothesize the approximation ratio of PCD for general odd $n$.

\end{abstract}
\section{Introduction} 
\subsection{Setting}
We consider the problem of locating a facility in a metric space.
A set of $n$ strategic agents have ideal locations for the facility $x_i$, where the cost for each agent is exactly the distance from the facility to their ideal location. Each agent prefers to be as close as possible to the facility. We are interested in placing the facility in a way that optimizes some  objective function. For example minimizing the social cost - The sum of costs for all agents.
A mechanism is a function that maps the reported ideal locations of all agents, to a single location for the facility. A central authority, receives the reported ideal locations of all agents, and uses a mechanism to place the facility to optimize the objective function. Since agents want to be as close as possible to the facility, there is no guarantee their reports will be truthful. Hence the central authority is interested in mechanisms that optimize a social utility function, while being \textit{strategy-proof} (i.e agents cannot benefit from misreporting their preferred location).\newline We can generalize these mechanisms to be randomized, thus mapping the agents ideal locations to a distribution over the metric space.
Evidently, the mechanism that places the facility on the optimal location will not always be strategy-proof. Hence, we are interested in  the cost approximation ratio of the mechanism.
The facility location problem can be varied to different metric spaces (Euclidean space , circle, graphs etc.), different objective functions, ranging number of facilities, and other constraints.

\subsection{Current Knowledge}

 The most basic case of facility location is on a line, where  \citet{black48} showed that selecting the median agent location is both optimal and strategy proof. \citet{0a97e823-e66c-383d-81ae-9d347515857e} further extended this idea and characterized all strategy proof mechanisms on a discrete/continuous line under general single-peaked preferences. \citet{dokow2012mechanism} characterized the class of mechanisms for discrete lines under symmetric preferences, which is richer than the general single-peaked class, since there are less ways to manipulate, and hence more strategyproof mechanisms. It is worth noting that the median mechanism can also be extended to any acyclic graph.
By contrast, a circle is essentially the simplest metric space where median based mechanisms fail to work. This fact draws attention to finding the best strategyproof mechanism on a circle, which remains unknown. \citet{SCHUMMER2002405} showed that any deterministic (onto) strategyproof mechanism on a circle must be dictatorial, which is obviously far from optimal. Therefore, the focus has shifted to the realm of randomized strategy proof mechanisms, which remains largely unexplored. Until recently, literature doesn't mention mechanisms that approximate the optimal social cost better than the Random Dictator mechanism, which gives an approximation ratio of $2-\frac{2}{n}$ for any metric space \citet{alon2009} . \citet{meir2019} presented two mechanisms (PCD and q-QCD) that improve the upper bound for 3 agents on a circle, and improved the lower bound slightly by use of linear programming. \citet{rogowski2025}  improved the upper bound to 7/4  by mixing the PCD and Random Dictator mechanism for any odd $n\geq 3$, and hypothesize the bound for this mechanism can be lowered to 1.5.

\subsection{Contribution and Method}
Our contribution is lowering the upper bound for five agents on a circle by finding the approximation ratio of the Proportional Circle Distance (PCD) mechanism  to be $7-4\sqrt2 \approx1.343$ and showing that this bound is tight. This further improves the best known bound given by Random Dictator of $\frac{8}{5}$. Moreover we hypothesize a closed formula for the approximation ratio of PCD for any odd $n$ (Figure \ref{fig:numerical}) . A comparison of the known and new results can be seen in Table \ref{tab:results}. 

\def\vl{\vphantom{\frac{2^{2^2}}{2_{2_2}}}}
\begin{table}[t]
\centering
\resizebox{\textwidth}{!}{
\begin{tabular}{|l|c|c|c|}
\hline
\textbf{Mechanism} & \textbf{$n=3$} & \textbf{$n=5$} & \textbf{odd $n\geq7$} \\ \hline
RD & $\frac{4}{3}\approx1.333\vl$ & $\frac{8}{5}=1.6$ & $2-\frac{2}{n}$ \citep{alon2009} \\ \hline
PCD & $\frac{5}{4}=1.25\vl$ \citep{meir2019} & $7-4\sqrt{2}\approx1.34$ [Main result]* & 
\begin{tabular}[c]{@{}c@{}}
$\geq 2-O(\frac{1}{\sqrt{n}})\vl$ \cite{meir2019} \\ 
 $\frac{2\big[(n^{2}-3n+4)-\sqrt{2}\,(n-1)^{3/2}\big]}{(n-3)^{2}}\vl$ [Hypothesis]* \\ 
 
\end{tabular} \\ \hline
PCD+RD & $\leq1.75\vl$ & $\leq 1.75$ & 
\begin{tabular}[c]{@{}c@{}}
$\leq 1.75\vl$ \citep{rogowski2025} \\ 
 $\leq1.5\vl$ [Hypothesis] \\ 
 
\end{tabular}  \\ \hline
Best Upper Bound & $\frac{7}{6}\approx1.166 \vl$ (QCD) \citet{meir2019} &  $7-4\sqrt{2}\approx1.34$ & $ \frac{7}{4}=1.75$ \\ \hline
 Best Lower Bound & $\vl$ 1.0456 \citep{meir2019} & ? & ? \\ \hline
\end{tabular}}\caption{Known and new results for approximation of strategyproof mechanisms on a circle. New results marked with asterisk.
\
\label{tab:results}}

\end{table}
\noindent To find the approximation ratio for PCD, we theoretically have to compute the maximum value of the ratio between the expected cost of the mechanism solution and the cost of the optimal position, for all possible instances. This of course is a very demanding task, due to the endless asymmetric instances of 5 agents on a circle, the fact that there isn't a compact way of determining the optimal location, and the fact that this ratio isn't a differentiable function. To overcome this, we partition our instance space, and try solving the problem regionally. By first finding bounds for cases that have 2 pairs of agents that share a location, we then attempt bounding the challenging asymmetric cases by showing the two pair cases actually give worse ratios. We thus gradually cover the instance space by using bounds we have already found, until we are left with a region where the approximation ratio for a given instance can be written explicitly, thus allowing us to use basic optimization techniques. In addition we utilize the fact that in most cases, the optimal location is on the agent facing the longest arc.
To find the worst instance and observe the general PCD mechanism behavior we created an application that visualizes instances neatly:
\href{https://chatgpt.com/canvas/shared/68e8b99e57348191bfd9b42e1da5091b}{Interactive App}.

\begin{figure}
    \centering
    \includegraphics[width=0.7\linewidth]{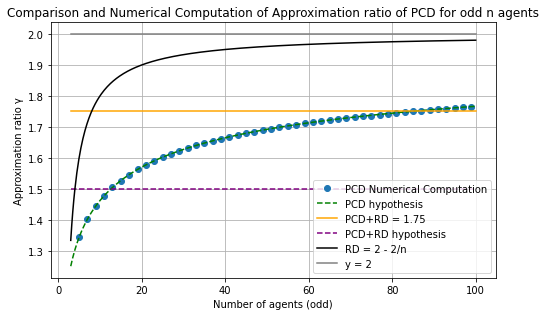}
    \caption{The hypothesized approximation ratio of PCD as number of agents grows, compared to a numerical computation, and other known mechanisms.}
    \label{fig:numerical}
\end{figure}

\section{Prelimineries}\label{sec:prelim}
We denote the metric space by $(X,d)$, where $d(x,y)$ is the shortest distance between $x$ and $y$, and we denote the ideal locations of the $n$ strategic agents by $x_i$ for $1\leq\ i\leq n$. An instance will be denoted as $\mathbf{x}=(x_1,x_2,..x_n)\in X^n$, and the facility location will be denoted by $a\in X$. The cost of placing the facility at $a$ for agent $i$ is exactly $c_i(a):=d(x_i,a)$, and the social cost is given by $ SC(\mathbf{x},a) =\sum_{i\in N} c_i(a)$. The mechanism will be denoted by $g:X^n \rightarrow X$. For randomized mechanisms $g:X^n \rightarrow \Delta X$, where $\Delta X$ is the set of distributions over $X$. 

\noindent We denote the randomized outcome by $h\in \Delta X$ and generalize the cost to 
$$c_i(h):= \E_{z\sim h} [d(x_i,z)].$$
The social cost generalizes naturally due to linearity of expectation.
\noindent A  mechanism $g$ is \textit{strategyproof} if for any report $\textbf{x}\in X^n$, for any player $i\in N$, and for any alternative report $x_i'$ we have:

$$c_i(g(\textbf{x})) \leq c_i(g(\textbf{x}')),$$ where we denoted $\textbf{x}'=(\mathbf{x}_{-i},x'_i)$ as the identical report to $\textbf{x}$ but $x_i'$ is placed instead of $x_i$. This definition simply means an agent cannot improve his cost by misreporting. \newline

\noindent On a circle, the distance between two points is the shortest arc length between them. The PCD (Proportional Circle Distance) mechanism selects each reported location $x_i$ with probablity proprtional to the length $L_i$ of the arc facing agent $i$. Formally:
The PCD mechanism assigns the facility to each agent's location with probability $P_i=\frac{L_i}{\sum_{i<n}L_i}$.
For the sake of simplicity, we study the unit circle, where the probability of placing the facility on $x_i$ is exactly $L_i$. This mechanism is strategy proof for every odd $n$ as proved by \citet{meir2019}. Thus, for this work we will always say \textbf{an agent's reported ideal location is simply his location}, since we are only working with a strategy proof mechanism.
\newline
\newline
 For the PCD mechanism for n agents, the social cost is explicitly given by $SC(\textbf{x},g(\textbf{x}))= \sum_{i=1}^nC_iP_i$, where $P_i$ denotes the probability of the facility being placed on agent $i$ and $C_i:=SC(\textbf{x},x_i)$ is social cost when placing the facility on agent $i$.
\newline 
We denote $\mathrm{OPT}(\textbf{x}) = \inf_{a \in X} \mathrm{SC}(\textbf{x},a)$ as the optimal facility location for profile $\textbf{x}$.
\noindent The ratio between the expected cost of a mechanism  solution $g(\textbf{x)}$ and the cost of the optimal cost is given by:
\[
        \gamma(g,\mathbf x) \ \defeq \  \frac{\mathrm{SC}(\mathbf{x}, g(\mathbf{x}))}{ \mathrm{SC}(\mathbf{x}, \mathrm{OPT}(\mathbf{x}))} \]
We often will write $\gamma(\textbf{x})$ when $g$ is clear from context. Consequently, the approximation ratio is given by:
\[\gamma(g):= \sup_{\textbf{x} \in X^n }\gamma(g,\mathbf x)\] 

 \noindent Finding $\gamma(g)$ is equivalent to maximizing $\gamma(\textbf{x)}$ for all $\textbf{x} \in X^n$.

\noindent In this work, The function $g$ always denotes the PCD mechanism. \newline
We will always index the agents in clockwise manner, and we will usually assume agent 3 gives $OPT(\textbf{x})$. We can do so since if 3 doesn't give $OPT(\textbf{x})$ we can shift the indices.
\newline For the PCD mechanism on a circle, the arc lengths between agents reported ideal location in a given instance give us a vector $\textbf{P}:= (P_1,P_2,P_3,P_4,P_5)$ with components that sum to 1.
\newline We  note that for every given instance $\textbf{x}$, up to a rotation of $0\leq \theta\leq 2\pi$ there is a corresponding $\textbf{P}$ which describes the instance equivalently. Hence we will often use $\textbf{x}$ and $\textbf{P}$ interchangeably.

\section{Results}

\subsection{Optimum Characterization}
We first state three Lemmas regarding the optimal location of a facility for $n$ agents on a circle. We first show that there exists an agent with an optimal location, then we show a useful trait of this agent, and lastly we prove that in some cases, we can easily identify which agent has optimal location. Proofs of these Lemmas are found in the Appendix. \newline

\begin{lemma}\label{lemma:opt_on_top}
For all instances $\textbf{x}$, there exists agent $i$ such that $x_i$ is optimal location.
\end{lemma}

\begin{lemma}\label{lemma:median_optimal} 

There exists an agent on optimal location, such that after assigning ties (other agents on optimal location and agents antipodal to optimal location), has $\frac{n-1}{2}$ agents clockwise from him and $\frac{n-1}{2}$ counterclockwise.. We call this agent \textit{median optimal}.
\end{lemma}

\begin{lemma}\label{lemma:large_arc_opt}
    
If $P_i\geq0.5$ for some agent $i$, then $x_i$ is optimal location. \end{lemma}

\noindent The main result we are looking to prove is the following theorem:
\newline
\begin{theorem}\label{theorem:theorem1} 
 The approximation ratio of the PCD mechanism for five agents on a circle is $\gamma = 7-4\sqrt{2}$
\end{theorem}

\noindent \textbf{From this point forward we denote $\alpha :=7-4\sqrt{2}$, which is approximately 1.34.}

\noindent As mentioned in the introduction, we are essentially trying to solve an optimization problem, of a function which is clearly not differentiable, and behaves in a rather complex manner due to the metric on the circle. To prove Theorem \ref{theorem:theorem1}, we initially find a candidate for the worst instance and the corresponding approximation ratio value $\alpha$. From numerical computation we expect the worst instance to be an instance that has two pairs of agents in different locations. This is analyzed in subsection \ref{subsec:2p} "Two Pair Analysis" . Afterwards, in subsection  \ref{subsec:largest_arc} "Instances with Large Arc", we show that every instance that contains an arc greater than 0.5 has an approximation ratio that is bounded by $\alpha$. Lastly, we tightly bound the social cost in subsection \ref{subsec:BSC} "Bounding the Social Cost". This helps us bound instances with a sufficiently high optimal cost. In total, these steps enable us to significantly reduce the "size" of the instance space, and solve the optimization problem in a domain where the function is differentiable and has an explicit form - allowing us to use standard optimization techniques.

\subsection{Two Pair Analysis}\label{subsec:2p}
\noindent We first find an instance where this bound is obtained:

\noindent In the following theorem we find a bound for $\gamma(\textbf{x)}$ for all instances that have 2 pairs of agents, where the agents in each pair share the same location. Instances that are made up of 2 pairs can be characterized by using 2 variables only, which makes computations much easier. By numerical computation, the instance that maximizes $\gamma(\textbf{x}$) is of this kind, hence this is a good place to start. Our main goal afterwards will be showing every other instance is better than the worst 2 pair instance. We give a proof sketch here, and the full proof can be found in the Appendix.

\begin{theorem}\label{Two Pair Analysis}: 
For all instances $\textbf{x} $ that have 2 pairs of agents sharing the same location, $\gamma(\textbf{x})\leq \alpha$.
\end{theorem}
\begin{proofsketch}
By fixing $x_3$ as the "lone agent" in a given instance, and denoting his distance from $x_1,x_2$ by $s$ and his distance from $x_4,x_5$ by $t$, we see all instances of this kind are of the form $P=(t,0,1-t-s,0,s)$. For each possible order of the lengths of the arcs, we can easily compute all the $C_i$ values, and moreover we can compute $OPT(\textbf{x})$ for each case - Thus allowing us to compute $\gamma(\textbf{x)}$ explicitly as $\gamma(t,s)$. We are left with an optimization problem of a 2 variable differentiable rational function with constraints. By a quick analysis, all of these functions obtain their maximum values on their boundary, and are easily computed. We find the maximum value to be $\alpha=7-4\sqrt{2}$, and the worst instance to be of the form $P=(0.207,0,0.5,0,0.293)$ which is seen in Figure~\ref{fig:combined_figure}(a). 
\end{proofsketch}

\begin{figure}[htbp]
    \centering
    \begin{subfigure}[b]{0.48\textwidth}
        \centering
        \raisebox{0.1cm}{\includegraphics[width=\textwidth]{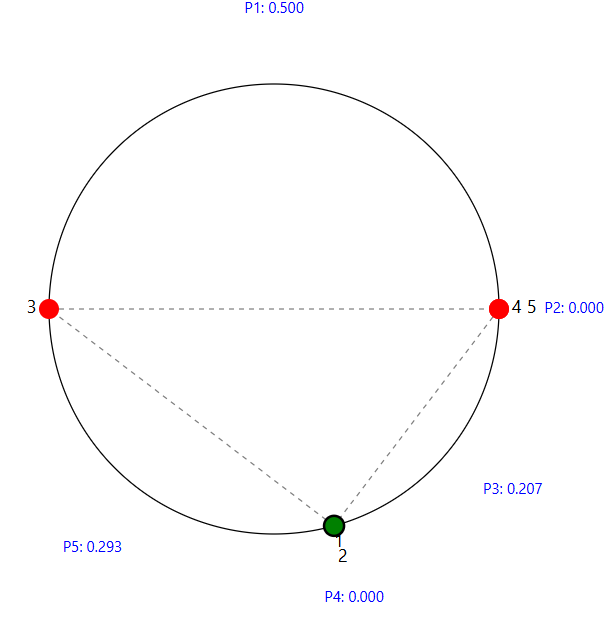}} 
        \caption{The worst 2 pair instance. Agent 3 is lone agent, and the green node represents the optimal location.}
        \label{fig:subim1}
    \end{subfigure}
    \hfill 
    \begin{subfigure}[b]{0.44\textwidth}
        \centering
        \raisebox{0cm}{\includegraphics[width=\textwidth]{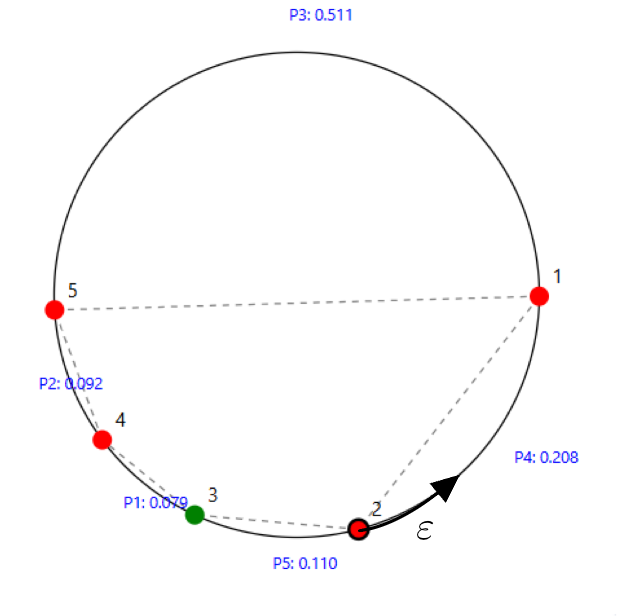}}
        \caption{An instance with an arc larger than 0.5. The anotation represents agent 2 movement counterclockwise by $\epsilon$.}
        \label{fig:subim2}
    \end{subfigure}
    \caption{}
    \label{fig:combined_figure}
\end{figure}

\noindent From this point forward, we will show that $\gamma(x)\leq \alpha$ for all instances. Before moving forward, we will take care of cases that have 3 agents on same location:
\newline \begin{lemma}\label{3_agents_same_spot}

 If \textbf{$x$} is an instance where 3 agents share the same location, and $P_3$ then $\gamma(\textbf{x})=1$.\end{lemma}
\begin{proof}
 Assume $x_2=x_3=x_4$. Hence, the arc lengths are: $P=(0,P_2,P_3,P_4,0)$. Moreover we know $C_2=C_3=C_4 \coloneqq C_{cluster}$. The Social cost of PCD is given by $SC(\textbf{x},g(\textbf{x})) = \sum_{i=1}^5 C_iP_i= C_{cluster}(P_2+P_3+P_4)=C_{cluster}$. Optimal location is obviously at the cluster, hence $\gamma(\textbf{x})= \frac{C_{cluster}}{C_{cluster}}=1$. \end{proof}

\subsection{Instances with Large Arc}\label{subsec:largest_arc}
\noindent Since we know exactly where the optimal location is for instances that have an arc greater or equal to 0.5, we will bound all these cases.

\begin{theorem}\label{large_arc_theorem}
For every instance $\bf x$ with an arc larger or equal to 0.5,  $\gamma(\bf{x}) \leq \alpha$.
\end{theorem} 
\begin{proof}

 Assume $P_3\geq0.5$. From Lemma \ref{lemma:large_arc_opt} we know that agent $3$ is at optimal location. We will show that in this case, moving $x_2$ (or symmetrically $x_4$) in a certain direction results in monotonic behavior of $\gamma(\textbf{x})$. By using this fact, we will move the agents in such a way that raises $\gamma(\textbf{x})$ until we reach a 2 pair instance, which we know is bounded by $\alpha$.
\newline
Lets assume we are not in a 2 pair instance. hence, either agent $2$ or agent $4$ (or both) are isolated at their location. Let's assume agent $2$ is isolated. We now move agent $2$ by epsilon counterclockwise: $x_2\rightarrow x_2+\epsilon$. This results in $(P_1,P_2,P_3,P_4,P_5) \rightarrow (P_1,P_2,P_3,P_4-\epsilon,P_5+\epsilon)$. Consequently, the costs change as follows: $(C_1,C_2,C_3,C_4,C_5)\rightarrow (C_1-\epsilon,C_2+2\epsilon,C_3+\epsilon,C_4+\epsilon,C_5+\epsilon)$. This step is easier visualized as seen in Figure \ref{fig:combined_figure}(b). Now, We can compute the new Social Cost denoted  
$SC' = \sum_{i=1}^5 C_i'P_i'$ \newline
= $P_1(C_1-\epsilon)+P_2(C_2+2\epsilon)+P_3(C_3+\epsilon)+(P_4-\epsilon)(C_4+\epsilon)+(P_5+\epsilon)(C_5+\epsilon)$
\newline = $SC+\epsilon(-P_1+2P_2+P_3+P_4-C_4+P_5+C_5)$
\newline = $SC+\epsilon A$ where $A$ is simply a constant.
The important thing here is that terms of $\epsilon^2$ canceled eachother out.
Now we define $\gamma(\epsilon) = \frac{SC'}{C_3'}= \frac{SC+\epsilon A}{C_3+\epsilon}$. \newline
Now $\gamma(\epsilon)$ is a one variable function of $\epsilon$, so we can differentiate it:\newline
$\gamma^{'}(\epsilon)= \frac{A(C_3+\epsilon)-(SC+\epsilon A)}{(C_3+\epsilon)^2} = \frac{AC_3-SC}{(C_3+\epsilon)^2}$. Since the numerator is constant and the denominator is always positive, we conclude $\gamma$ is a monotone function of $\epsilon$. Thus moving agent $2$ in one direction leads to nondecreasing $\gamma$, and the other direction leads to nonincreasing $\gamma$. The exact same argument holds for moving agent $4$. Thus lets move agent $2$ and $4$ in the directions where $\gamma$ doesnt increase, until they reach another agent. We reach an instance that is either 2 pair, or an instance that has 3 agents on same location, and by Theorem \ref{Two Pair Analysis} and Lemma \ref{3_agents_same_spot} these cases are bounded by $\alpha$ as wished, and hence our original instance is bounded by $\alpha$ aswell. \end{proof}

\subsection{Bounding the Social Cost}\label{subsec:BSC}
\noindent The following theorem  gives a tight global bound for the social cost. Bounding the social cost helps bound $\gamma(\textbf{x})$, thus making this result beneficial. The complete proof can be found in the Appendix, but the main idea of the proof will be shown here.
\newline
 \begin{theorem}\label{bounding_sc}
    
 $SC(\textbf{x},g(\textbf{x}))= \sum_{i=1}^5 C_iP_i \leq 1.2$ for any instance $\textbf{x}$, and 1.2 is obtained when all agents are distanced equally from one another. \end{theorem}

\begin{proofsketch}

  The social cost is given by $\sum_{i=1}^5 C_iP_i$. Given an instance of agents on the circle, $\textbf{P}=(P_1,P_2,P_3,P_4,P_5)$ are all known. However, each $C_i$ value is determined by summing the distances from agent $i$ to the rest of the agents, for example:
  \[
\begin{aligned}
C_1 = d(x_1,x_2)+d(x_1,x_3)+d(x_1,x_4)+d(x_1,x_5)
\end{aligned}
\] Since $d(x_i,x_j)= min(|x_i-x_j|, 1-|x_i-x_j|) \leq |x_i-x_j|$, we can bound each $C_i$ value by some linear combination of $P_i$'s. Hence, we can bound each $C_iP_i$ value by a second degree polynomial of $P_i$ variables, and consequently after summing, we can bound the whole social cost by a second degree polynomial with $P_i$ variables. By  finding the maximal value of this polynomial under the constraint of $\sum_{i=1}^5P_i=1$, we bound the social cost tightly as wished.
  
\end{proofsketch}

\noindent We now utilize this useful bound to significantly reduce our instance space:
\begin{corollary}\label{cor_opt<0.9}

For every instance\textbf{ x} with $0.9\leq OPT(\textbf{x})$, $\gamma(\textbf{x})\leq  \alpha$ \end{corollary}
\begin{proof}

$\gamma(\textbf{x})=\frac{SC(\textbf{x},g(\textbf{x}))}{OPT(\textbf{x})}\leq \frac{1.2}{OPT(\textbf{x})}\leq \alpha$ \end{proof} 

\noindent We now know that every instance that has an arc length $P_i\geq 0.5$ or  $OPT(\textbf{x})>0.9$ is bounded by $\alpha$.

\begin{lemma}\label{lemma9}    
For instance $\textbf{x}$, If $x_3$ is optimal location, agent $3$ is median optimal and $P_1+P_5\geq 0.5$, then $\gamma(\textbf{x})\leq \alpha$. \end{lemma}
\begin{proof}

By Lemma~\ref{lemma:median_optimal},  $d(x_3,x_4)$ and $d(x_3,x_5)$ are calculated clockwise, and $d(x_3,x_2)$ and $d(x_3,x_1)$ are calculated counterclockwise. Thus $d(x_3,x_4)=P_1$ and $d(x_3,x_2)= P_5$. Moreover $d(x_3,x_5)\geq P_1$ and $d(x_3,x_1)\geq P_5$. By calculating $C_3$ we obtain $ OPT(x)=C_3\geq 2(P_1+P_5)\geq1$, and by Corollary \ref{cor_opt<0.9}, $\gamma(\textbf{x}) \leq 1.34$.\end{proof}

\subsection{Main Result}
\noindent We are now ready to prove that for every instance $\textbf{x}$ $\gamma(\textbf{x})\leq \alpha$.
To find $\gamma$, we are essentially solving an optimization problem with constraints: We are maximizing $\gamma(\textbf{x})$, under the constraints
$ \sum_{i=1}^5P_i=1$ and $P_i\geq 0$ for all $i$.
\newline Moreover it is sufficient to look at instances where $x_3$ is optimal location, and agent $3$ is median optimal. That is because given an instance $\textbf{x}$ we can always relabel the agents to achieve so. By using Theorem \ref{large_arc_theorem} and Lemma~\ref{lemma9}, we can restate our optimization problem to be:\newline
Maximizing $\gamma(\textbf{x})$ under the constraints $ \sum_{i=1}^5P_i=1$ , $P_i \geq 0$ for all $i$, $P_i\leq 0.5$ for all $i$ and $P_1+P_5\leq 0.5$. This is a optimization probelm on a significantly "smaller" instance space. It is important to note we are interested in maximizing this function only for instances where agent $3$ is median optimal. Instances that don't achieve this aren't in our function's domain. Moreover, as can be seen by the constraints, the domain of $\gamma$ is compact, and since $\gamma$ is continuous, it attains a maximum in the domain. The proof involves many different cases, which are solved using similar techniques, thus only part of the proof will be shown here, the rest can be found in the Appendix.
\begin{theorem}\label{main_result}
For all instances $\textbf{x}$, $\gamma(\textbf{x})\leq \alpha$.\end{theorem}

\begin{proof}
 First we assume agent $3$ is median optimal. We can calculate the costs: \newline
$C_1=2P_4+P_5+d(x_1,x_5)+d(x_1,x_4)= 2P_4+2P_5+P_3+d(x_1,x_4)$ \newline
$C_2=P_5+P_4 +d(x_2,x_4)+d(x_2,x_5)= P_5+P_4+P_5+P_1+d(x_2,x_5)$\newline
$C_3=P_2+2P_1+2P_5+P_4$\newline
$C_4=P_2+P_1+d(x_2,x_4)+d(x_4,x_1)=P_2+P_1+P_1+P_5+d(x_4,x_1)$\newline
$C_5=2P_2+P_1+d(x_5,x_2)+d(x_5,x_1)=2P_2+P_1+P_3+d(x_5,x_2)$\newline

\begin{figure}
    \centering
    \includegraphics[width=0.5\linewidth]{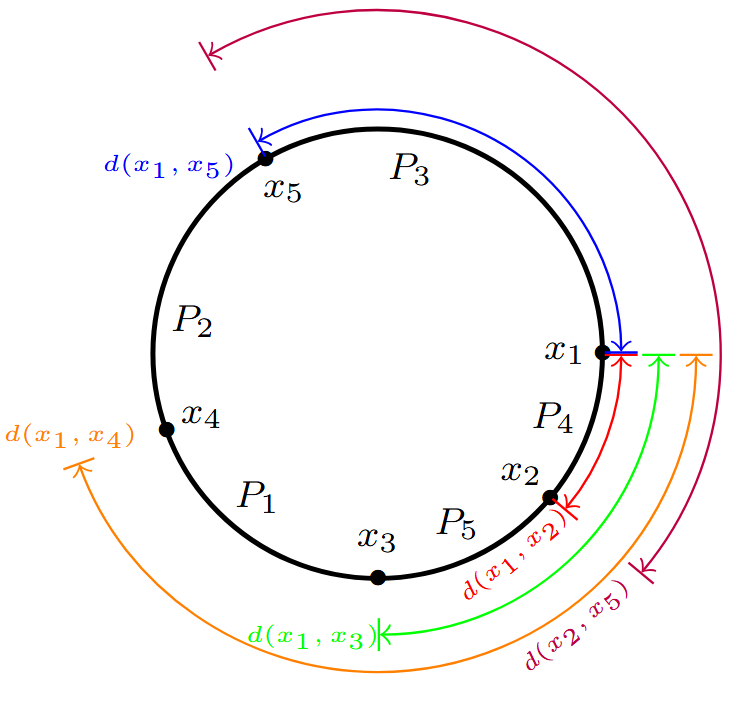}
    \caption{An example of an instance where $d(x_1,x_4)=d_A(x_1,x_4)=P_1+P_5+P_4$ and $d(x_2,x_5)=d_D(x_2,x_5)=P_3+P_4$ and thus an instance for $\gamma_{AD}$. Moreover, all distances from $x_1$ are shown, which demonstrates the ability to calculate $C_1$.}
    \label{fig:main}
\end{figure}

\noindent Where we used the fact that $3$ is median optimal to calculate distances, $d(x_1,x_5)=P_3$ (since $P_3\leq 0.5$) and $d(x_2,x_4)=P_1+P_5$ (since $P_1+P_5\leq 0.5$). \newline
At this stage, all of the $C_i$'s can be explicitly written  if $d(x_1,x_4)$ and $d(x_2,x_5)$ are determined. If they are determined, we can write all the costs as functions of $P_i$, and thus we can write $\gamma(\textbf{x})$ as a function of $P_i$'s only. Using the shortest arc metric we denote:
$d_A(x_1,x_4)= P_1+P_5+P_4$ , $d_B(x_1,x_4)=P_3+P_2$, $d_C(x_2,x_5)=P_5+P_1+P_2$ and $d_D(x_2,x_5)=P_3+P_4$.
We denote $\gamma_{AC}$, $\gamma_{AD}$, $\gamma_{BC}$, $\gamma_{BD}$, by the corresponding distances included in each $\gamma$ (See Figure \ref{fig:main} for example). By the symmetry of the problem, $\gamma_{AD}$ and $\gamma_{BC}$ are identical optimization problems, so we only need to solve one of them. By solving the appropriate optimization problem for $\gamma_{AC}$,$\gamma_{AD}$ and $\gamma_{BD}$ we will bound the original $\gamma(\textbf{x})$. \newline
Our method will be showing that each of the functions above has no interior maxima, and hence finds its maxima on the boundary. Since we previously handled the boundary cases, we will achieve our goal.\newline
To simplify notation we denote $(P_1.P_2,P_3,P_4,P_5)=(a,b,c,d,e)$.\newline We first solve the optimization problem for $\gamma_{AC}$: \newline
$\gamma_{AC}=\frac{a(a+c+3d+2e)+b(3e+d+2a+b)+c(b+2a+2e+d)+d(b+3a+2e+d)+e(3b+2a+c+e)}{b+2a+2e+d}$ \newline by inserting the constraint $\sum_{i=1}^5P_i=1$ we can simplify\newline
$\gamma_{AC}(a,b,d,e)=
\frac{-2a^{2}-2ab+2ad-2ae+2be-2de-2e^{2}+3a+b+d+3e}{\,2a+b+d+2e\,}$.\newline
\[
\text{Let }\gamma:= \gamma_{AC}(a,b,d,e),\quad D:=2a+b+d+2e
\]
\[
\nabla \gamma =
\begin{bmatrix}
\dfrac{-4a-2b+2d-2e+3-2\gamma}{D}\\[6pt]
\dfrac{-2a+2e+1-\gamma_{}}{D}\\[6pt]
\dfrac{2a-2e+1-\gamma_{}}{D}\\[6pt]
\dfrac{-2a+2b-2d-4e+3-2\gamma_{}}{D}
\end{bmatrix}=\begin{bmatrix}
0\\
0\\
0\\
0
\end{bmatrix}.
\]
\newline  By summing $\frac{\partial \gamma_{}}{\partial b}=0$ and $\frac{\partial \gamma_{}}{\partial d}=0$ we obtain $\gamma_{}=1$ and we also conclude $a=e$. By summing $\frac{\partial \gamma_{}}{\partial a}=0$ and $\frac{\partial \gamma_{}}{\partial e}=0$ we obtain $a=e=\frac{1}{3}$, which means $a+e>0.5$ which is outside our boundaries. We conclude there are no interior maxima.
\newline The maxima must lay on the boundary. If it lays on $P_i=0.5$ for some $i$, we are done, since that case is bounded by $\alpha$ by Theorem \ref{large_arc_theorem}. If it lays on $P_1+P_5=0.5$, we are done, since that case is bounded by $\alpha$ by Lemma \ref{lemma9}. Lastly, the maxima may lay on the boundary given by $P_i\geq0$ for some $i$.
\newline Let's assume maxima is on a=0. Thus we will find maxima of: \newline
\[
\gamma_{AC}(0,b,d,e)=
\frac{2be-2de-2e^{2}+b+d+3e}{\,b+d+2e\,}
\].
\newline 
\[
\text{Let }\gamma_{AC}:=\gamma_{AC}(0,b,d,e),\quad D:=b+d+2e.
\]
\[
\nabla\gamma =
\begin{bmatrix}
\dfrac{2e+1-\gamma}{D}\\[6pt]
\dfrac{1-2e-\gamma}{D}\\[6pt]
\dfrac{2b-2d-4e+3-2\gamma}{D}
\end{bmatrix}_{(b,d,e)}=\begin{bmatrix}
0\\
0\\

0
\end{bmatrix}
\]
\newline By summing $\frac{\partial\gamma}{b}=0$ and $\frac{\partial\gamma}{d}=0$ we obtain $\gamma=1$ and $e=0$ - which lies on boundary, and hence there is no interior maxima. By the same reasoning as before, the maxima lies on the boundary. If it lies on $P_i=0.5$ or $P_1+P_5=0.5$ we are done. Otherwise, it lies on $P_i=0$, which means $P_1=0$ and $P_{i\neq1}=0$ which is either a 2 pair instance or an instance with 3 agents on same location, which is already bounded by Theorem \ref{Two Pair Analysis} and Lemma \ref{3_agents_same_spot}. In the same fashion, we continue for the cases of $b=0$, $c=0$, $d=0$ and $e=0$. After so, we do the same for $\gamma_{AD}$ and $\gamma_{BD}$, which com the proof. \end{proof}
\subsection{Beyond Five Agents}
By looking at the worst instances $n=3$, $n=5$, and $n\geq 7$ by numerical analysis, we notice a "clustering" phenemona, where $k$ agents report one location, another $k$ agents report a different location, and the last single agent is located exactly antipodal to one of the previous clusters - portrayed in Figure \ref{fig:clustering}. By assuming the worst instance for all odd $n$ is of this form, we can find a closed formula for the approximation ratio for the PCD mechanism. By denoting $t$ as the distance between the two clusters, one can find the approximation ratio $\gamma(k,t)=  \frac{kt - t + \tfrac{1}{2}}{\,2kt - 2kt^{2} - t + \tfrac{1}{2}\,}$. Assuming $k$ is a constant, we can maximize $\gamma(t)$ and obtain:
\noindent
\[
\begin{gathered}
t_{\max}(k)=\frac{\sqrt{k}-1}{2(k-1)}\\[6pt]
\gamma_{hypothesis}(k)=\frac{2k^{2}-2k^{3/2}-k+1}{(k-1)^{2}}\\[6pt]
\gamma_{hypothesis}(n)=\frac{2\big[(n^{2}-3n+4)-\sqrt{2}\,(n-1)^{3/2}\big]}{(n-3)^{2}}
\end{gathered}
\]
 We clearly see that the approximation of PCD approaches 2 (\citet{meir2019}) when $n\rightarrow \infty$, and furthermore we see that $t_{max} \rightarrow0$ as $k\rightarrow\infty$, which shows that the worst instance where $n$ is large has essentially all of the agents in one cluster and a lone agent exactly antipodal.
\newline

\begin{figure}[h]
    \centering
    \includegraphics[width=0.25\linewidth]{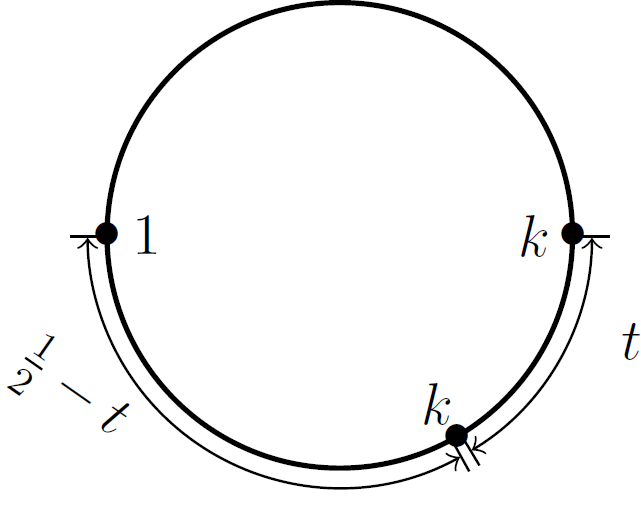}
    \caption{The "clustering" phenomena - two clusters with $k$ agents, one lone agent, and an arc of length 0.5}
    \label{fig:clustering}
\end{figure}

\subsection{Interactive App}
To better understand how the PCD mechanism works and analyze it's behavior for various instances, we built an interactive application which visualizes instances. The app calculates in real time the $P_i$ and $C_i$ values, shows the optimal location, and allows you to move agents freely on the circle (without overlap). The social cost and $\gamma(\textbf{x)}$ are shown aswell. Simultaneous movement of 2 nodes is also possible. We were interested in simultaneous movement since it is possible to move 2 nodes in a such a way that keeps the value $OPT(\textbf{x)}$ the same, and $\gamma(\textbf{x)}$ changes only as a function of the social cost. The app can be found here: \href{https://chatgpt.com/canvas/shared/68e8b99e57348191bfd9b42e1da5091b}{Interactive App}.

\section{Discussion} 
We considered a strategyproof facility location on a circle with 5 agents, and showed the approximation ratio for the PCD mechanism is $7-4\sqrt2 \approx1.34$. This lowers the previous upper bound given by the Random Dictator mechanism of $\frac{8}{5}$. In the context of finding the best strategyproof mechanism on a circle, the $n=3$ case is bounded by $\frac{7}{6}$ by the q-QCD mechanism (\citet{meir2019})  and the recent result by \citet{rogowski2025} bounded the general odd $n\geq 3$ case by $\frac{7}{4}$ using the mixed RD+PCD mechanism. This recent result only beats the RD mechanism when $n\geq7$, thus our result in some sense fills the gap, and is currently the best known result for 5 agents. 
After noticing the clustering phenomena is present in the worst instances of the PCD mechanism for $n\geq7$ by numerical computation, we hypothesized a closed formula for the approximation ratio for general odd $n$. We leave as open question to prove this is in fact the case. Only for $n\geq85$ does RD+PCD mixed mechanism beat the hypothesized approximation ratio for PCD, and hence proving so remains interesting.
In addition, we leave as an open question whether there is a better strategyproof mechanism than PCD for 5 agents on a circle.

\nocite{*}
\bibliographystyle{plainnat}   
\bibliography{references}   

\clearpage
\appendix
\section{Appendix}
\begin{lemma_restatement}[\ref{lemma:opt_on_top}]
For all instances $\textbf{x}$, there exists agent $i$ such that $x_i$ is optimal location.
\end{lemma_restatement}

\begin{proof}
For a given instance $\textbf{x}$, lets unwrap it to the $[0,1)$ with circular arithmetic. Let's denote $F(z)=\sum_{i=1}^nd(x_i,z)$, the function that measures the social cost of placing the facility at $z$. Denote $f_i(z)=d(x_i,z)$. Since distances are measured by minimal arc length, $f_i(z)=\min(|z-x_i|, 1-|z-x_i|)$. For each $f_i$, other than the point $x_i$ and its antipodal point, $f_i^{'}$ is piecewise constant, and equal to 1 for $(x_i,x_i+0.5)$ and equals -1 for $(x_i-0.5,x_i)$. Thus we conclude $F^{'}(z)$ (where it's defined) is piecewise constant, and $F^{'}(z)$ jumps up by +2 on points with agents, and jumps down by -2 on antipodal points of agents. Now let's look at the interval bounded by 2 adjacent agents $[x_k, x_{k+1}]$. Beside the boundary, there are no agents in this region, but there may be antipodal points of agents in this region, and hence, $F^{'}(z)$ is non decreasing, and hence concave. Since concave functions on compact sets recieve there minimum on the boundary of the domain, we conclude that $\min_{z\in[x_k, x_{k+1}]} F(z)=\min(F(x_k), F(x_{k+1})$ By repeating this logic for all arcs (a finite amount), we conclude the minimum value obtained is on an endpoint, and thus $OPT(\textbf{x})$ is obtained by at least one of the agents. 
\end{proof}

\begin{proof}[Alternative proof of Lemma~\ref{lemma:opt_on_top}]
Let's look at an instance $\textbf{x}$ and a point $x$ which is not on an agent, and assume $OPT(\textbf{x)}$ is obtained on $x$. We will show this point can't be optimal. Denote $L$= number of agents in open semicircle clockwise from x,  $R$ = number of agents in open semicircle counterclockwise from x and $A$= number of agents antipodal to x. First let's look at the case where $A$=0. In this case, since $n$ is odd, $L\neq R$ and hence moving $x$ by $\epsilon$ either left or right will result in lowering OPT(\textbf{x}), making $x$ not optimal. If $A$ is even, the above proof also holds. Otherwise, let's assume $A$= k is odd.
Hence, $n-k=R + L$ is even. If $L=R$ than obviously moving $x$ by epsilon in either direction will lower $OPT(\textbf{x)}$ making $x$ not optimal due to the antipodal agents. Otherwise $L\neq R$ and their difference must be even. Since $k$ is odd, by moving $x$ either left or right by $\epsilon$, we will definitely decrease $OPT(\textbf{x})$ since the odd $k$ agent can't make up for the even difference of $L$ and $R$. Hence $x$ can't be optimal, and thus optimal is obtained on one of the agents.
\end{proof}

\begin{lemma_restatement}[\ref{lemma:median_optimal}] 

There exists an agent on optimal location, such that after assigning ties ( other agents on optimal location and agents antipodal to optimal location), has $\frac{n-1}{2}$ agents clockwise from him and $\frac{n-1}{2}$ counterclockwise.. We call this agent \textit{median optimal}.
\end{lemma_restatement}
\begin{proof}
We denote L,R,A as in previous lemma, and denote $B$=number of agents on optimal location. Obviously we have $A+B+L+R=n$, and since $n$ is odd, we denote $n=2k+1$. If we move the facility $\epsilon$ clockwise, the SC rises (due to optimality) by $\epsilon(R-L+B-A)\geq 0$, and by moving the facility $\epsilon$ counterclockwise, the SC rises by $\epsilon(L-R+B-A)\geq 0$.
Hence we conclude $|L-R|\leq B-A$. Let's look at the case where there is one lone agent on optimum location and no agents antipodal: $B=1$ and $A=0$, $L+R=n-1=2k$, and hence $|L-R|$ is even and at most one, and thus is zero, and $L=R=k$, as we wished.
Denote $T=A+B-1$ for all the agents that are tied. Our goal now is to take the T "tied" agents and assign them to L and R accordingly, such that in the end, an agent on the optimal location will have $k$ agents clockwise from him and $k$ counterclockwise. First we claim $|L-R|\leq T$ . If $A\geq 1$ then: $T=A+B-1\geq B\geq B-A \geq |L-R|$. If $A=0$, $T=B-1$ and $|L-R|\leq B$. But $L-R$ and $L+R$ have same parity, and due to $T=A+B-1=2k-(L+R)$, we know $T$ and $L-R$ have same parity. $|L-R| < B$ and thus $|L-R| \leq B-1 = T$.
Now we we take $|L-R|$ agents from $T$ and assign them to $min(R,L)$. Now $R$ and $L$ are equal and we are left with $T-|L-R|$ agents to assign. Because $T$ and $|L-R|$ have the same parity, their difference is even and hence we can continue assigning the remaining evenly between the two sides, resulting in $k$ agents on each side of an optimal agent. \end{proof}

\begin{proof}[Alternative proof of Lemma ~\ref{lemma:median_optimal} for 5 agents]
    
Let's assume $x_3$ is optimal location. If agent $3$ has the above trait, we are done. Otherwise lets assume w.log that the distances between $x_3$ and $x_1, x_2, x_5$ are calculated counterclockwise. We also assume that if the distance between 2 nodes is zero, the distance between them is clockwise iff they are indexed accordingly. 
\newline
$C_3 = d(x_3,x_4) + d(x_3,x_2)+d(x_3,x_1)+d(x_3,x_5)$
\newline $=d(x_3,x_4)+3d(x_3,x_2)+d(x_2,x_1)+d(x_2,x_5)$
\newline $\geq d(x_2,x_4)+2d(x_3,x_2)+d(x_2,x_1)+d(x_2,x_5)$
\newline $\geq d(x_2,x_4)+d(x_3,x_2)+d(x_2,x_1)+d(x_2,x_5) = C_2$
\newline where we used the fact that $x_1,x_2,x_3,x_4$ are on same semi circle, and $d(x_2,x_4) \leq d(x_2,x_3)+d(x_3,x_4)$. If $d(x_3,x_2)$ is not zero, we get a strict inequality, and that is a contradiction to $C_3$ being optimal location. Hence, $d(x_3,x_2)=0$, which results in agent 2 having 2 agent's distance calculated clockwise (3,4) and 2 counterclockwise (1,5), making him median optimal.
\end{proof}

\begin{lemma_restatement}[\ref{lemma:large_arc_opt}]
    
If $P_i\geq0.5$ for some agent $i$, then $x_i$ is optimal location. \end{lemma_restatement}

\begin{proof}

Since there can only be one arc greater than 0.5, we can assume w.log that $P_k\geq 0.5$ where k is the the median agent when indexing clockwise. Hence all of the agents lie on the same semi circle. By the previous lemma, at the optimal location, there is an agent that computes his distance to other agents $\frac{n-1}{2}$ clockwise and $\frac{n-1}{2}$ counterclockwise. Since all agents lie in same semicircle, only agent $k$ can have this trait, and thus $x_k$ is optimal location. \end{proof}

\medskip
\begin{theorem_restatement}[\ref{Two Pair Analysis}] 
For all instances $\textbf{x} $ that have 2 pairs of agents sharing the same location, $\gamma(\textbf{x})\leq \alpha$.
\end{theorem_restatement}

\begin{proof}
    
 We assume w.log that $x_3$ is the "lone agent" - we can do so since if this isn't the case, we can relabel the agents and achieve so, and have the same instance up to a rotation. We denote his distance from $x_1,x_2$ by $s$ and his distance from $x_4,x_5$ by $t$. Consequently, the arc between the 2 pairs is of length $1-t-s$. Thus, all the instances are of the form $P=(t,0,1-t-s,0,s)$. 
\newline To analyze these instances, we need to somehow define an order on the length of the arcs. w.log we can assume $s\leq t$. Moreover, we are interested in 2 cases: The first is where the largest arc is in front of the lone agent, and the other being the case where the largest arc is in front of one of the pairs. This reduces  to orders: $A: 1-t-s\geq t \geq s$ and $B: t\geq 1-t-s\geq s$. Moreover, we denote $A_1$ the case where $1-t-s\geq 0.5$ and $A_2$ the case where $1-t-s\leq 0.5$. Similarly we denote $B_1$ the case where $t\geq 0.5$ and $B_2$ the case where $t\leq 0.5$.
Now, using these inequalities alone, we can explicitly compute $C_i$ for all $i$ in each case (See figure below). In cases, $A_1$, $B_1$ and $B_2$ the optimal node can easily be identified by using the inequalities above and Lemma 3. In the $A_2$ case, we can only conclude that either $C_1$ or $C_3$ is optimal cost.
\newline
\begin{center}
    
\end{center}
\begin{tabular}{|p{6cm}|p{6cm}|}
\hline
\textbf{$A_1$} & \textbf{$A_2$} \\
\hline
$C_1=s+2(s+t)$) \newline
$C_2=s+2(s+t)$ \newline
$C_3=2s+2t$ \newline
$C_4=t+2(s+t)$ \newline
$C_5=t+2(s+t)$

&

$C_1=2(1-t-s)+s$ \newline
$C_2=2(1-t-s)+s$ \newline
$C_3=2t+2s$ \newline
$C_4=2(1-t-s)+t$\newline
$C_5=2(1-t-s)+t$

\\
\hline
\textbf{$B_1$} & \textbf{$B_2$} \\
\hline
$C_1=2(1-t-s)+s$ \newline
$C_2=2(1-t-s)+s$ \newline
$C_3=2(1-t)+2s$ \newline
$C_4=2(1-t-s)+(1-t)$\newline
$C_5=2(1-t-s)+(1-t)$\newline

&
$C_1= s+2(1-t-s)$ \newline
$C_2= s+2(1-t-s)$ \newline
$C_3 =2s+2t$ \newline
$C_4=t+2(1-t-s)$ \newline
$C_5=t+2(1-t-s)$ \newline

\\
\hline
\end{tabular}
\newline
\newline
Since we know $C_i$ and $P_i$ for all $i$, we can write $\gamma(x)$ explicitly and analyze it for each case. We will have 5 different functions ($A_2$ has 2 options) and we will find the maximum value of $\gamma(x)$ in each case. In each case, $\gamma(x)$ is differentiable, and continious on a compact set. Hence attains a maximum. In all 5 cases, no interior extremum points were found, and all maximum bounds were found on the boundaries, namely when the largest arc was exactly $0.5$. 
\newline
\newline
\begin{tabular}{|p{6cm}|p{6cm}|}
\hline
\textbf{$A_1$} & \textbf{$A_2$} \\
\hline
 
$\frac{2ts+2s+2t}{2s+2t}$ 
Max = $9/8$

&

$\frac{-4s^2-4t^2-6ts+4t+4s}{2-2t-s}$ Max = $9/8$ \newline
$\frac{-4s^2-4t^2-6ts+4t+4s}{2t+ts}$ Max = $9/8$

\\
\hline
\textbf{$B_1$} & \textbf{$B_2$} \\
\hline
$\frac{-4s^2-4st+3s-2t+2}{2-2t-s}$ Max = $\alpha$

&
$\frac{-4s^2-4t^2-6ts+4t+4s}{2-2t-s}$ Max = $\alpha$

\\
\hline
\end{tabular}
\newline We conclude that for 2 pair instances, $\gamma(x)$ is bounded by $\alpha$. Moreover, we conclude the worst 2 pair instance is given by $P= (\frac{-1+\sqrt{2}}{2},0,0.5,0,\frac{2-\sqrt{2}}{2} )\approx (0.207,0,0.5,0,0.293)$. \end{proof}

\begin{theorem_restatement}[\ref{bounding_sc}]
    
 $SC(\textbf{x},g(\textbf{x}))= \sum_{i=1}^n C_iP_i \leq 1.2$ for any instance $\textbf{x}$, and 1.2 is obtained when all agents are distanced equally from one another. \end{theorem_restatement}

\begin{proof}
    
By using the fact that distances are measured by smallest arc between points, and the fact that the sum of the arcs is equal to one, we find the following inequalities:
\[
\begin{aligned}
C_1 = d(x_1,x_2)+d(x_1,x_3)+d(x_1,x_4)+d(x_1,x_5) &\le 2P_3 + 2P_4 + P_5 + P_2 &&= 1 - P_1 + P_3 + P_4 \\
C_2 = d(x_2,x_1)+d(x_2,x_3)+d(x_2,x_4)+d(x_2,x_5) &\le 2P_5 + 2P_4 + P_3 + P_1 &&= 1 - P_2 + P_4 + P_5 \\
C_3 = d(x_3,x_1)+d(x_3,x_2)+d(x_3,x_4)+d(x_3,x_5) &\le 2P_1 + 2P_5 + P_4 + P_2 &&= 1 - P_3 + P_1 + P_5 \\
C_4 = d(x_4,x_1)+d(x_4,x_2)+d(x_4,x_3)+d(x_4,x_5) &\le 2P_1 + 2P_2 + P_3 + P_5 &&= 1 - P_4 + P_1 + P_2 \\
C_5 = d(x_5,x_1)+d(x_5,x_2)+d(x_5,x_3)+d(x_5,x_4) &\le 2P_3 + 2P_2 + P_1 + P_4 &&= 1 - P_5 + P_3 + P_2
\end{aligned}
\]
\newline Let's multiply each $C_i$ by $P_i$ to obtain: 
\newline
\[
\begin{aligned}
C_1P_1 &\leq P_1 - P_1^2 + P_3P_1 + P_4P_1\\
C_2P_2 &\leq P_2 - P_2^2 + P_4P_2 + P_5P_2\\
C_3P_3 &\leq P_3 - P_3^2 + P_1P_3 + P_5P_3\\
C_4P_4 &\leq P_4 - P_4^2 + P_1P_4 + P_2P_4\\
C_5P_5 &\leq P_5 - P_5^2 + P_3P_5 + P_2P_5
\end{aligned}
\]
Now by adding together all the inequalities, we obtain a bound for the social cost:\newline
$SC(\textbf{x},g(\textbf{x}))\leq 1-P_1^2-P_2^2-P_3^2-P_4^2-P_5^2 + 2P_1P_3+2P_4P_1+2P_4P_2+2P_5P_2+2P_5P_3$ \newline
To simplify this bound, we use the fact that the arc lengths sum to 1: \newline
$SC(\textbf{x},g(\textbf{x}))\leq 1-P_1^2-P_2^2-P_3^2-P_4^2 -(1-P_1-P_2-P_3-P_4)^2+2P_1P_3+2P_4P_1+2P_4P_2 +2(P_2+P_3)(1-P_1-P_2-P_3-P_4)$ \newline
On the right hand side we have a quadratic 4 variable function we'll denote $f$. If we bound it, we bound the social cost. We'll find the maximum value of $f$ by computing $\nabla f = 0$: \newline
$\nabla f = (4P_1+4P_2+2P_3-2, 4P_1+8P_2+6P_3+2P_4-4,2P_1+6P_2+8P_3+4P_4-4,2P_2+4P_3+4P_4-2) = (0,0,0,0)$ \newline
This results in the following system of linear equations:
\[
\begin{bmatrix}
4 & 4 & 2 & 0 \\
4 & 8 & 6 & 2 \\
2 & 6 & 8 & 4 \\
0 & 2 & 4 & 4
\end{bmatrix}
\begin{bmatrix}
P_1 \\ P_2 \\ P_3 \\ P_4
\end{bmatrix}
=
\begin{bmatrix}
2 \\ 4 \\ 4 \\ 2
\end{bmatrix}
\]
The solution is $(0.2,0.2,0.2,0.2)$, and by inserting this back to $f$ we get $SC(\textbf{x},g(\textbf{x}))\leq 1.2$. We also notice that the social cost is highest when all the agents are equally apart from one another. \end{proof}

\begin{theorem_restatement}[\ref{main_result}]
For all instances $\textbf{x}$, $\gamma(\textbf{x})\leq \alpha$.\end{theorem_restatement}

\begin{proof}
 First we assume agent $3$ is median optimal. We can calculate the costs: \newline
$C_1=2P_4+P_5+d(x_1,x_5)+d(x_1,x_4)= 2P_4+2P_5+P_3+d(x_1,x_4)$ \newline
$C_2=P_5+P_4 +d(x_2,x_4)+d(x_2,x_5)= P_5+P_4+P_5+P_1+d(x_2,x_5)$\newline
$C_3=P_2+2P_1+2P_5+P_4$\newline
$C_4=P_2+P_1+d(x_2,x_4)+d(x_4,x_1)=P_2+P_1+P_1+P_5+d(x_4,x_1)$\newline
$C_5=2P_2+P_1+d(x_5,x_2)+d(x_5,x_1)=2P_2+P_1+P_3+d(x_5,x_2)$\newline
Where we used the fact that $3$ is median optimal to calculate distances, $d(x_1,x_5)=P_3$ (since $P_3\leq 0.5$) and $d(x_2,x_5)=P_1+P_5$ (since $P_1+P_5\leq 0.5$). \newline
At this stage, all of the $C_i$'s can be explicitly written  if $d(x_1,x_4)$ and $d(x_2,x_5)$ are determined. If they are determined, we can write all the costs as functions of $P_i$, and thus we can write $\gamma(\textbf{x})$ as a function of $P_i$'s only. Using the shortest arc metric we denote:
$d_A(x_1,x_4)= P_1+P_5+P_4$ , $d_B(x_1,x_4)=P_3+P_2$, $d_C(x_2,x_5)=P_5+P_1+P_2$ and $d_D(x_2,x_5)=P_3+P_4$.
We denote $\gamma_{AC}$, $\gamma_{AD}$, $\gamma_{BC}$, $\gamma_{BD}$, by the corresponding distances included in each $\gamma$. By the symmetry of the problem, $\gamma_{AD}$ and $\gamma_{BC}$ are identical optimization problems, so we only need to solve one of them. By solving the appropriate optimization problem for $\gamma_{AC}$,$\gamma_{AD}$ and $\gamma_{BD}$ we will bound the original $\gamma(\textbf{x})$. \newline
Our method will be showing that each of the functions above has no interior maxima, and hence finds it's maxima on the boundary. Since we previously handled the boundary cases, we will achieve our goal.\newline
To simplify notation we denote $(P_1.P_2,P_3,P_4,P_5)=(a,b,c,d,e)$.\newline We first solve the optimization problem for $\gamma_{AC}$: \newline
$\gamma_{AC}=\frac{a(a+c+3d+2e)+b(3e+d+2a+b)+c(b+2a+2e+d)+d(b+3a+2e+d)+e(3b+2a+c+e)}{b+2a+2e+d}$ \newline by inserting the constraint $\sum_{i=1}^5P_i=1$ we can simplify\newline
$\gamma_{AC}(a,b,d,e)=
\frac{-2a^{2}-2ab+2ad-2ae+2be-2de-2e^{2}+3a+b+d+3e}{\,2a+b+d+2e\,}$.\newline
\[
\text{Let }\gamma:= \gamma_{AC}(a,b,d,e),\quad D:=2a+b+d+2e
\]
\[
\nabla \gamma =
\begin{bmatrix}
\dfrac{-4a-2b+2d-2e+3-2\gamma}{D}\\[6pt]
\dfrac{-2a+2e+1-\gamma_{}}{D}\\[6pt]
\dfrac{2a-2e+1-\gamma_{}}{D}\\[6pt]
\dfrac{-2a+2b-2d-4e+3-2\gamma_{}}{D}
\end{bmatrix}=\begin{bmatrix}
0\\
0\\
0\\
0
\end{bmatrix}.
\]
\newline  By summing $\frac{\partial \gamma_{}}{\partial b}=0$ and $\frac{\partial \gamma_{}}{\partial d}=0$ we obtain $\gamma_{}=1$ and we also conclude $a=e$. By summing $\frac{\partial \gamma_{}}{\partial a}=0$ and $\frac{\partial \gamma_{}}{\partial e}=0$ we obtain $a=e=\frac{1}{3}$, which means $a+e>0.5$ which is a out of our boundaries. We conclude there are no interior maxima.
\newline The maxima must lay on the boundary. If it lays on $P_i=0.5$ for some $i$, we are done, since that case is bounded by $\alpha$ by Theorem \ref{large_arc_theorem}. If it lays on $P_1+P_5=0.5$, we are done, since that case is bounded by $\alpha$ by Lemma \ref{lemma9}. Lastly, the maxima may lay on the boundary given by $P_i\geq0$ for some $i$.
\newline Let's assume maxima is on a=0. Thus we will find maxima of: \newline
\[
\gamma_{AC}(0,b,d,e)=
\frac{2be-2de-2e^{2}+b+d+3e}{\,b+d+2e\,}
\].
\newline 
\[
\text{Let }\gamma_{AC}:=\gamma_{AC}(0,b,d,e),\quad D:=b+d+2e.
\]
\[
\nabla\gamma =
\begin{bmatrix}
\dfrac{2e+1-\gamma}{D}\\[6pt]
\dfrac{1-2e-\gamma}{D}\\[6pt]
\dfrac{2b-2d-4e+3-2\gamma}{D}
\end{bmatrix}_{(b,d,e)}=\begin{bmatrix}
0\\
0\\

0
\end{bmatrix}
\]
\newline By summing $\frac{\partial\gamma}{b}=0$ and $\frac{\partial\gamma}{d}=0$ we obtain $\gamma=1$ and $e=0$ - which lies on boundary, and hence there is no interior maxima. By the same reasoning as before, the maxima lies on the boundary. If it lies on $P_i=0.5$ or $P_1+P_5=0.5$ we are done. Otherwise, it lies on $P_i=0$, which means $P_1=0$ and $P_{i\neq1}=0$ which is either a 2 pair instance or an instance with 3 agents in same location, which is already bounded by Theorem \ref{Two Pair Analysis} and Lemma \ref{3_agents_same_spot}.
\newline Let's assume maxima is on $b=0$. Thus we will find maxima of:
\[
\gamma_{AC}(a,0,d,e)=
\frac{-2a^{2}+2ad-2ae-2de-2e^{2}+3a+d+3e}{\,2a+d+2e\,}
\]
\newline
\[
\text{Let }\gamma:=\gamma_{AC}(a,0,d,e),\quad D:=2a+d+2e.
\]
\[
\nabla \gamma =
\begin{bmatrix}
\dfrac{-4a+2d-2e+3-2\gamma}{D}\\[6pt]
\dfrac{\,2a-2e+1-\gamma\,}{D}\\[6pt]
\dfrac{-2a-2d-4e+3-2\gamma}{D}
\end{bmatrix}_{(a,d,e)}=\begin{bmatrix}
0\\
0\\
0
\end{bmatrix}
\]
\newline 

This linear system gives: $d=0.5-3a$, $e=7a-1$ and $\gamma=3-12a$. Inserting this back into $\gamma_{AC}(a,0,d,e)$ we find $a\approx 0.09$ or $a\approx 0.13$ which both lead to $e<0$ which is out of boundary. Hence there are no interior points. By the same reasoning as before, the maxima lies on the boundary. If it lies on $P_i=0.5$ or $P_1+P_5=0.5$ we are done. Otherwise, it lies on $P_i=0$, which means $P_2=0$ and $P_{i\neq2}=0$ which is either a 2 pair instance or an instance with 3 agents on same location, which is already bounded by Theorem \ref{Two Pair Analysis} and Lemma \ref{3_agents_same_spot}.
\newline
Let's assume maxima is on $d=0$. Thus we will find maxima of:
\[
\gamma_{AC}(a,b,0,e)=
\frac{-2a^{2}-2ab-2ae+2be-2e^{2}+3a+b+3e}{\,2a+b+2e\,}
\]
\[
\text{Let }\gamma:=\gamma_{AC}(a,b,0,e),\quad D:=2a+b+2e.
\]
\[
\nabla\gamma =
\begin{bmatrix}
\dfrac{-4a-2b-2e+3-2\gamma}{D}\\[6pt]
\dfrac{-2a+2e+1-\gamma}{D}\\[6pt]
\dfrac{-2a+2b-4e+3-2\gamma}{D}
\end{bmatrix}_{(a,b,e)}=\begin{bmatrix}
0\\
0\\

0
\end{bmatrix}.
\]

The linear system gives $b=0.5-3e$, $a=7e-1$, $\gamma=3-12e$. Inserting this back into $\gamma_{AC}(a,b,0,e)$ gives $e\approx0.09$ or $e\approx 0.13$ which results in $a<0$ which is out of boundaries.  Hence there are no interior points. By the same reasoning as before, the maxima lies on the boundary. If it lies on $P_i=0.5$ or $P_1+P_5=0.5$ we are done. Otherwise, it lies on $P_i=0$, which means $P_4=0$ and $P_{i\neq4}=0$ which is a either a 2 pair instance or an instance with 3 agents on same location, which is already bounded by Theorem \ref{Two Pair Analysis} and Lemma \ref{3_agents_same_spot}.
\newline
Lets assume maxima is on $e=0$. Thus we will find maxima of: 
\[
\gamma_{AC}(a,b,d,0)=
\frac{-2a^{2}-2ab+2ad+3a+b+d}{\,2a+b+d\,}
\]
\[
\text{Let }\gamma:=\gamma_{AC}(a,b,d,0),\quad D:=2a+b+d.
\]
\[
\nabla\gamma =
\begin{bmatrix}
\dfrac{-4a-2b+2d+3-2\gamma}{D}\\[6pt]
\dfrac{-2a+1-\gamma}{D}\\[6pt]
\dfrac{2a+1-\gamma}{D}
\end{bmatrix}_{(a,b,d)}=\begin{bmatrix}
0\\
0\\

0
\end{bmatrix}
\]
By summing $\frac{\partial\gamma}{b}=0$ and $\frac{\partial\gamma}{d}=0$ we obtain $\gamma=1$ and $a=0$ - which lies on the boundary, hence there is not interior point.
By the same reasoning as before, the maxima lies on the boundary. If it lies on $P_i=0.5$ or $P_1+P_5=0.5$ we are done. Otherwise, it lies on $P_i=0$, which means $P_5=0$ and $P_{i\neq5}=0$ which is either a 2 pair instance or an instance with 3 agents on same location, which is already bounded by Theorem \ref{Two Pair Analysis} and Lemma \ref{3_agents_same_spot}.
\newline
Lastly, we need to take care of the case of the boundary $P_3=0$. To do so, we will rewrite $\gamma_{AC}$ as a function of $(a,b,c,d)$ and insert $c=0$.
\[
\gamma_{AC}(a,b,c,d)
=\frac{-2a^{2}-6ab-2ac+2ad+2a-4b^{2}-6bc-4bd+4b-2c^{2}-2cd+c+1}
{\,2 - b - 2c - d\,}.
\]
\[
\gamma_{AC}(a,b,0,d)=
\frac{-2a^{2}-6ab+2ad+2a-4b^{2}-4bd+4b+1}{\,2-b-d\,}
\]
\[
\text{Let }\gamma:=\gamma_{AC}(a,b,0,d),\quad D:=2-b-d.
\]
\[
\nabla\gamma =
\begin{bmatrix}
\dfrac{-4a-6b+2d+2}{D}\\[6pt]
\dfrac{-6a-8b-4d+4+\gamma}{D}\\[6pt]
\dfrac{2a-4b+\gamma}{D}
\end{bmatrix}_{(a,b,d)}=\begin{bmatrix}
0\\
0\\

0
\end{bmatrix}.
\]

The linear equations give:
\[
a + b = \tfrac{1}{2}, \qquad d = b, \qquad \gamma = 6b - 1.
\]

By inserting back to $\gamma_{AC}(a,b,0,d)$:
\[
\frac{-6b^{2} + 2b + \tfrac{3}{2}}{\,2 - 2b\,} = 6b - 1.
\]

\[
b^{*}=d^{*}=1-\frac{\sqrt{15}}{6}, \qquad
a^{*}=e^{*}=\frac{\sqrt{15}-3}{6}, \qquad
c^{*}=0, \qquad
\gamma^{*}=5-\sqrt{15} \approx1.12
\]
which is bounded by 1.34. If the maxima does lie on boundary, like before we reduce to a 2 pair instance which is bounded.
\newline
\newline
\textbf{We conclude $\gamma_{AC}$ is bounded by $\alpha$.}
\newline
Next we optimize the $\gamma_{AD}$ function. The method is exactly the same as $\gamma_{AC}$, hence the computation will be more concise: \newline
$\gamma_{AD}= \frac{a(3d+2e+c+a)+b(2e+2d+a+c)+c(b+2a+2e+d)+d(b+3a+2e+d)+e(2b+a+2c+d)}{b+2a+2e+d}$
By using constraint:
\[
\gamma_{AD}(a,b,d,e)
=\frac{-2a^{2}-4ab+2ad-4ae-2b^{2}-2be-2de-4e^{2}+3a+2b+d+4e}{2a+b+d+2e}.
\]
\[
\text{Let }\gamma:=\gamma_{AD}(a,b,d,e),\quad D:=2a+b+d+2e.
\]
\[
\nabla\gamma =
\begin{bmatrix}
\dfrac{3-4a-4b+2d-4e-2\gamma}{D}\\[6pt]
\dfrac{2-4a-4b-2e-\gamma}{D}\\[6pt]
\dfrac{1+2a-2e-\gamma}{D}\\[6pt]
\dfrac{4-4a-2b-2d-8e-2\gamma}{D}
\end{bmatrix}_{(a,b,d,e)}=\begin{bmatrix}
0\\
0\\
0\\
0
\end{bmatrix}.
\]

\noindent The linear equations give:
\newline $a=d=\frac{4\gamma-1}{22}$, $b=\frac{7-6\gamma}{22}$, $e=\frac{10-7\gamma}{22}$.
Inserting these back to $\gamma_{AD}$ results in $b<0$, hence there is no interior point.
Now we check the boundaries for maxima. If maxima lies on boundary $P_i=0.5$ or $P_1+P_5=0.5$ we are done by Theorem \ref{large_arc_theorem} and Lemma \ref{lemma9}.
\newline Assume $a=0$.
\[
\gamma_{AD}(0,b,d,e)
=\frac{-2b^{2}-2be-2de-4e^{2}+2b+d+4e}{\,b+d+2e\,}.
\]
\[
\text{Let }\gamma:=\gamma_{AD}(0,b,d,e),\quad D:=b+d+2e.
\]
\[
\nabla\gamma =
\begin{bmatrix}
\dfrac{2-4b-2e-\gamma}{D}\\[6pt]
\dfrac{1-2e-\gamma}{D}\\[6pt]
\dfrac{4-2b-2d-8e-2\gamma}{D}
\end{bmatrix}_{(b,d,e)}=\begin{bmatrix}
0\\
0\\

0
\end{bmatrix}
\]
By looking at $\frac{\partial\gamma_{AC}}{\partial d}=0$ we obtain $\gamma\leq1$ which is either bounded or not an interior point. So the maxima lays on boundary.
 If it lies on $P_i=0.5$ or $P_1+P_5=0.5$ we are done. Otherwise, it lies on $P_i=0$, which means $P_1=0$ and $P_{i\neq1}=0$ which is either a 2 pair instance or an instance with 3 agents on same location, which is already bounded by Theorem \ref{Two Pair Analysis} and Lemma \ref{3_agents_same_spot}.
\newline
Let's assume $b=0$.
\[
\gamma_{AD}(a,0,d,e)
=\frac{-2a^{2}+2ad-4ae-2de-4e^{2}+3a+d+4e}{\,2a+d+2e\,}.
\]

\[
\text{Let }\gamma:=\gamma_{AD}(a,0,d,e),\quad D:=2a+d+2e.
\]
\[
\nabla\gamma =
\begin{bmatrix}
\dfrac{3-4a+2d-4e-2\gamma}{D}\\[6pt]
\dfrac{1+2a-2e-\gamma}{D}\\[6pt]
\dfrac{4-4a-2d-8e-2\gamma}{D}
\end{bmatrix}_{(a,d,e)}=\begin{bmatrix}
0\\
0\\

0
\end{bmatrix}
\]

The linear equations give:
$a=\frac{2\gamma+1}{20}$, $d=\frac{4\gamma-3}{10}$, $e=\frac{11-8\gamma}{20}$
\newline Inserting these values into $\gamma$ and solving results in points outside of boundaries. 
So the maxima lays on boundary.
 If it lies on $P_i=0.5$ or $P_1+P_5=0.5$ we are done. Otherwise, it lies on $P_i=0$, which means $P_2=0$ and $P_{i\neq2}=0$ which is either a 2 pair instance or an instance with 3 agents on same location, which is bounded by Theorem \ref{Two Pair Analysis} and Lemma \ref{3_agents_same_spot}.
 \newline
 \newline
 Lets assume $d=0$:
 \[
\gamma_{AD}(a,b,0,e)
=\frac{-2a^{2}-4ab-4ae-2b^{2}-2be-4e^{2}+3a+2b+4e}{\,2a+b+2e\,}.
\]
\[
\text{Let }\gamma:=\gamma_{AD}(a,b,0,e),\quad D:=2a+b+2e.
\]
\[
\nabla\gamma =
\begin{bmatrix}
\dfrac{3-4a-4b-4e-2\gamma}{D}\\[6pt]
\dfrac{2-4a-4b-2e-\gamma}{D}\\[6pt]
\dfrac{4-4a-2b-8e-2\gamma}{D}
\end{bmatrix}_{(a,b,e)}=\begin{bmatrix}
0\\
0\\

0
\end{bmatrix}
\]

\noindent The linear equations give:\newline
$a=\gamma-\tfrac14$, $b=\tfrac12-\gamma$, $e=\tfrac12-\frac{\gamma}{2}$ \newline
Thus $b$ is negative, and hence there are no interior points. So the maxima lays on boundary.
 If it lies on $P_i=0.5$ or $P_1+P_5=0.5$ we are done. Otherwise, it lies on $P_i=0$, which means $P_4=0$ and $P_{i\neq4}=0$ which is either a 2 pair instance or an instance with 3 agents on same location, which is already bounded by Theorem \ref{Two Pair Analysis} and Lemma \ref{3_agents_same_spot}.
 \newline
 Let's assume $e=0$:
 \[
\gamma_{AD}(a,b,d,0)
=\frac{-2a^{2}-4ab+2ad-2b^{2}+3a+2b+d}{\,2a+b+d\,}.
\]
\[
\text{Let }\gamma:=\gamma_{AD}(a,b,d,0),\quad D:=2a+b+d.
\]
\[
\nabla\gamma =
\begin{bmatrix}
\dfrac{3-4a-4b+2d-2\gamma}{D}\\[6pt]
\dfrac{2-4a-4b-\gamma}{D}\\[6pt]
\dfrac{1+2a-\gamma}{D}
\end{bmatrix}_{(a,b,d)}=\begin{bmatrix}
0\\
0\\

0
\end{bmatrix}
\]

\noindent The linear equations give: \newline
$a=\frac{\gamma-1}{2}$, $b=1-\frac{3\gamma}{4}$, $d=\frac{\gamma}{2}-\frac{1}{2}$\newline
By inserting these values back to $\gamma$ we find $\gamma=\frac{4}{3}$ and $b=0$ which is out of boundary. 
Hence there are no interior points. So the maxima lays on boundary.
 If it lies on $P_i=0.5$ or $P_1+P_5=0.5$ we are done. Otherwise, it lies on $P_i=0$, which means $P_5=0$ and $P_{i\neq5}=0$ which is either a 2 pair instance or an instance with 3 agents on same location, which is bounded by Theorem \ref{Two Pair Analysis} and Lemma \ref{3_agents_same_spot}.
 \newline
 Lastly, we need to take care of the case of the boundary $P_3=0$. To do so, we will rewrite $\gamma_{AD}$ as a function of $(a,b,c,d)$ and insert $c=0$.
 \[
\gamma_{AD}(a,b,c,d)
=\frac{-2a^{2}-6ab-4ac+3a-4b^{2}-6bc-4bd+4b-4c^{2}-6cd+4c-2d^{2}+3d}{\,2-b-2c-d\,}.
\]
\[
\gamma_{AD}(a,b,0,d)
=\frac{-2a^{2}-6ab-4b^{2}-4bd-2d^{2}+3a+4b+3d}{\,2-b-d\,}.
\]
\[
\text{Let }\gamma:=\gamma_{AD}(a,b,0,d),\quad D:=2-b-d.
\]
\[
\nabla\gamma =
\begin{bmatrix}
\dfrac{3-4a-6b}{D}\\[6pt]
\dfrac{4-6a-8b-4d+\gamma}{D}\\[6pt]
\dfrac{3-4b-4d+\gamma}{D}
\end{bmatrix}_{(a,b,d)}=\begin{bmatrix}
0\\
0\\

0
\end{bmatrix}
\]

\noindent The linear equations give: \newline
$a=-\frac{3}{10}$, $b=\frac{7}{10}$, $\gamma=4d-\frac{1}{5}$
which is obviously out of boundaries. Hence there is no interior point and maxima lays on boudary.
If it lies on $P_i=0.5$ or $P_1+P_5=0.5$ we are done. Otherwise, it lies on $P_i=0$, which means $P_3=0$ and $P_{i\neq3}=0$ which is either a 2 pair instance or an instance with 3 agents on same location, which is already bounded by Theorem \ref{Two Pair Analysis} and Lemma \ref{3_agents_same_spot}.
\newline
\textbf{We conclude $\gamma_{AD}$ is bounded by 1.34.}
\newline
\newline
We now optimize the $\gamma_{BD}$ function:\newline
$\gamma_{BD}=\frac{a(b+2c+2d+e)+b(a+c+2d+2e)+c(b+2a+2e+d)+d(2a+2b+c+e)+e(a+2b+2c+d)}{b+2a+2e+d}$ \newline
\[
\gamma_{BD}(a,b,d,e)
= \frac{-4a^{2}-4ab-2ad-6ae-2b^{2}-2be-2d^{2}-4de-4e^{2}+4a+2b+2d+4e}{2a+b+d+2e}.
\]
\[
\text{Let }\gamma:=\gamma_{BD}(a,b,d,e),\quad D:=2a+b+d+2e.
\]
\[
\nabla\gamma =
\begin{bmatrix}
\dfrac{-8a-4b-2d-6e+4-2\gamma}{D}\\[6pt]
\dfrac{-4a-4b-2e+2-\gamma}{D}\\[6pt]
\dfrac{-2a-4d-4e+2-\gamma}{D}\\[6pt]
\dfrac{-6a-2b-4d-8e+4-2\gamma}{D}
\end{bmatrix}_{(a,b,d,e)}=\begin{bmatrix}
0\\
0\\
0\\
0
\end{bmatrix}.
\]

\noindent Solving the linear equations gives:
$a=b=d=e=\frac{2-\gamma}{10}$ and inserting this into $\gamma_{BD}$ gives $a=b=d=e=0$ which is on boundary.
Hence there are no interior points.
Inserting these back to $\gamma_{BD}$ results in $b<0$, hence there is no interior point.
Now we check the boundaries for maxima. If maxima lies on boundary $P_i=0.5$ or $P_1+P_5=0.5$ we are done by Theorem \ref{large_arc_theorem} and Lemma \ref{lemma9}.
\newline  Otherwise, Assume $a=0$:
\[
\gamma_{BD}(0,b,d,e)
= \frac{-2b^{2}-2be-2d^{2}-4de-4e^{2}+2b+2d+4e}{b+d+2e}.
\]
\[
\text{Let }\gamma:=\gamma_{BD}(0,b,d,e),\quad D:=b+d+2e.
\]
\[
\nabla\gamma =
\begin{bmatrix}
\dfrac{-4b-2e+2-\gamma}{D}\\[6pt]
\dfrac{-4d-4e+2-\gamma}{D}\\[6pt]
\dfrac{-2b-4d-8e+4-2\gamma}{D}
\end{bmatrix}_{(b,d,e)}\begin{bmatrix}
0\\
0\\

0
\end{bmatrix}
\]
By subtracting $\frac{\partial \gamma_{BD}}{\partial b}$ and $\frac{\partial \gamma_{BD}}{\partial d}$ we obtain $e=0$, which is on boundary. Hence there are no interior points, and maxima lies on boundary.
If it lies on $P_i=0.5$ or $P_1+P_5=0.5$ we are done. Otherwise, it lies on $P_i=0$, which means $P_1=0$ and $P_{i\neq1}=0$ which is either a 2 pair instance or an instance with 3 agents on same location, which is already bounded by Theorem \ref{Two Pair Analysis} and Lemma \ref{3_agents_same_spot}.
\newline
Let's assume $b=0$:
\[
\gamma_{BD}(a,0,d,e)
= \frac{-4a^{2}-2ad-6ae-2d^{2}-4de-4e^{2}+4a+2d+4e}{2a+d+2e}.
\]
\[
\text{Let }\gamma:=\gamma_{BD}(a,0,d,e),\quad D:=2a+d+2e.
\]
\[
\nabla\gamma =
\begin{bmatrix}
\dfrac{-8a-2d-6e+4-2\gamma}{D}\\[6pt]
\dfrac{-2a-4d-4e+2-\gamma}{D}\\[6pt]
\dfrac{-6a-4d-8e+4-2\gamma}{D}
\end{bmatrix}_{(a,d,e)}=\begin{bmatrix}
0\\
0\\

0
\end{bmatrix}
\]
Solving the linear equations gives:
\newline $a=\frac{2-\gamma}{6}$, $d=\frac{2-\gamma}{12}$, $e=\frac{2-\gamma}{12}$ \newline
And by inserting these values back to $\gamma_{BD}$ we get $a=d=e=0$ which is on boundary. Hence there are no interior points and maxima lies on boundary.
If it lies on $P_i=0.5$ or $P_1+P_5=0.5$ we are done. Otherwise, it lies on $P_i=0$, which means $P_2=0$ and $P_{i\neq2}=0$ which is either a 2 pair instance or an instance with 3 agents on same location, which is already bounded by Theorem \ref{Two Pair Analysis} and Lemma \ref{3_agents_same_spot}.
\newline
Let's assume $d=0$:
\[
\gamma_{BD}(a,b,0,e)
= \frac{-4a^{2}-4ab-6ae-2b^{2}-2be-4e^{2}+4a+2b+4e}{2a+b+2e}.
\]
\[
\text{Let }\gamma:=\gamma_{BD}(a,b,0,e),\quad D:=2a+b+2e.
\]
\[
\nabla\gamma =
\begin{bmatrix}
\dfrac{-8a-4b-6e+4-2\gamma}{D}\\[6pt]
\dfrac{-4a-4b-2e+2-\gamma}{D}\\[6pt]
\dfrac{-6a-2b-8e+4-2\gamma}{D}
\end{bmatrix}_{(a,b,e)}=\begin{bmatrix}
0\\
0\\

0
\end{bmatrix}
\]
Solving the linear equations gives: \newline
$a=\frac{2-\gamma}{12}$, $b=\frac{2-\gamma}{12}$, $e=\frac{2-\gamma}{6}$
\newline Inserting these values back to $\gamma_{BD}$ gives $a=b=e=0$ which is on boundary. Hence their are no interior points and maxima is on boundary.
If it lies on $P_i=0.5$ or $P_1+P_5=0.5$ we are done. Otherwise, it lies on $P_i=0$, which means $P_4=0$ and $P_{i\neq4}=0$ which is either a 2 pair instance or an instance with 3 agents on same location, which is already bounded by Theorem \ref{Two Pair Analysis} and Lemma \ref{3_agents_same_spot}.
\newline
Assume $e=0$:
\[
\gamma_{BD}(a,b,d,0)
= \frac{-4a^{2}-4ab-2ad-2b^{2}-2d^{2}+4a+2b+2d}{2a+b+d}.
\]
\[
\text{Let }\gamma:=\gamma_{BD}(a,b,d,0),\quad D:=2a+b+d.
\]
\[
\nabla\gamma =
\begin{bmatrix}
\dfrac{-8a-4b-2d+4-2\gamma}{D}\\[6pt]
\dfrac{-4a-4b+2-\gamma}{D}\\[6pt]
\dfrac{-2a-4d+2-\gamma}{D}
\end{bmatrix}_{(a,b,d)}=\begin{bmatrix}
0\\
0\\

0
\end{bmatrix}
\]
Solving the linear equations gives:
\newline
$a=\frac{2-\gamma}{6}$, $b=\frac{2-\gamma}{12}$, $d=\frac{2-\gamma}{6}$
Inserting these values into $\gamma_{BD}$ gives $a=b=d=0$, which is on boundary. Hence their are no interior points, and the maxima lies on the boundary.
If it lies on $P_i=0.5$ or $P_1+P_5=0.5$ we are done. Otherwise, it lies on $P_i=0$, which means $P_5=0$ and $P_{i\neq5}=0$ which is already bounded by Theorem \ref{Two Pair Analysis} and Lemma \ref{3_agents_same_spot}.
\newline
Lastly, we need to take care of the case of the boundary $P_3=0$. To do so, we will rewrite $\gamma_{BD}$ as a function of $(a,b,c,d)$ and insert $c=0$.
\[
\gamma_{BD}(a,b,c,d)
= \frac{-2a^{2}-4ab-2ac-4b^{2}-6bc-2bd-4c^{2}-4cd-2d^{2}+2a+4b+4c+2d}{2-b-2c-d}.
\]
\[
\gamma_{BD}(a,b,0,d)
= \frac{-2a^{2}-4ab-4b^{2}-2bd-2d^{2}+2a+4b+2d}{2-b-d}.
\]
\[
\text{Let }\gamma:=\gamma_{BD}(a,b,0,d),\quad D:=2-b-d.
\]
\[
\nabla\gamma =
\begin{bmatrix}
\dfrac{-4a-4b+2}{D}\\[6pt]
\dfrac{-4a-8b-2d+4+\gamma}{D}\\[6pt]
\dfrac{-2b-4d+2+\gamma}{D}
\end{bmatrix}_{(a,b,d)}=\begin{bmatrix}
0\\
0\\

0
\end{bmatrix}
\]
Solving the linear equations gives:
$a=\frac{1-\gamma}{6}$, $b=\frac{2+\gamma}{6}$, $d=\frac{2+\gamma}{6}$
\newline 
Inserting these values back to $\gamma_{BD}$ gives $a=e=0$ which is on boundary. Hence there are no interior points, and maxima lies on boundary.
If it lies on $P_i=0.5$ or $P_1+P_5=0.5$ we are done. Otherwise, it lies on $P_i=0$, which means $P_3=0$ and $P_{i\neq3}=0$ which is bounded by Theorem \ref{Two Pair Analysis} and Lemma \ref{3_agents_same_spot}.
\newline
\textbf{We conclude $\gamma_{BD}$ is bounded by 1.34.
}\newline
And finally, since $\gamma_{AC}$, $\gamma_{AD}$, $\gamma_{BD}$ are all bounded, we know for sure that for every instance $\textbf{x}$ our original $\gamma(\textbf{x)}$ is equal to one of the 3 bounded functions 
$\gamma_{AC}$, $\gamma_{AD}$, $\gamma_{BD}$, and hence $\gamma(\textbf{x)}$ is bounded for all $\textbf{x}$ as wished. \end{proof}

\noindent While trying to bound the instance space regionally, we attempted to bound cases that are "close" to the instance where all agents are equally spreaded apart on the circle. This result didn't contribute to the final solution, but does cover many cases.\newline

\begin{theorem}\label{theorem: equidistance}
Every instance $P$ such that $\forall i$ $|P_i-0.2|\leq \eps$ we have $\gamma(P)=\frac{1.2}{1.2-3\epsilon}\approx1+2.5\eps$.
\end{theorem} 
\begin{proof}

In the Equidistance instance, all agent's are on an optimal location. So lets take $C_3$ as $OPT(\textbf{x})$. Now, let's look at an instance where $P_i$ change by at most $\epsilon$, such that $C_3$ decreases as much as possible: This instance is: \newline $(0.2,0.2,0.2,0.2,0.2) \rightarrow (0.2-\epsilon,0.2+\frac{\epsilon}{2},0.2+\epsilon,0.2+\frac{\epsilon}{2}, 0.2-\epsilon)$. Intuitively we simply moved the agents as close to possible to agent 3, without changing the arc lengths by more than $\epsilon$.
This leads to $C_3=1.2 \rightarrow C_3-3\epsilon$.
Now, given an instance $P$ such that $|P_i-0.2|\leq\epsilon$: \newline
$\gamma(P)=\frac{\sum_{i=1}^5P_iC_i}{C_3}\leq \frac{1.2}{1.2-3\epsilon}\approx 1+2.5\epsilon$ where we used the fact that $SC(\textbf{x},g(\textbf{x}))\leq 1.2$ by Theorem \ref{bounding_sc}. \end{proof}

\begin{corollary}

Every instance $P$ such that $\forall i$ $|P_i-0.2|\leq0.1$ we have $\gamma(P)\leq 1.34$.
\end{corollary}
\begin{proof}

Inserting $\epsilon=0.1$ in Theorem \ref{theorem: equidistance} we obtain:
\newline $\gamma(P)\leq \frac{1.2}{1.2-3\epsilon}=\frac{1.2}{0.9}\leq 1.34$ \end{proof}
\newpage

\end{document}